\documentclass[12pt]{article}

\usepackage{a4wide}
\usepackage{amsfonts,amsmath,amsthm} 
\usepackage{enumerate}
\usepackage{slashed}

\title{(Generalized) Post Correspondence Problem and semi-Thue systems}
\author{Fran\c{c}ois Nicolas\thanks{E-mail address: \texttt{nicolas@cs.helsinki.fi}}}

\newtheorem{definition}{Definition}
\newtheorem{theorem}{Theorem}
\newtheorem{proposition}{Proposition}
\newtheorem{corollary}{Corollary}
\newtheorem{lemma}{Lemma}
\newtheorem{claim}{Claim}
\newtheorem{fact}{Fact}
\newtheorem{remark}{Remark}
\newtheorem{cex}{Counterexample}

\newcommand{\pbPCP}{\textup{PCP}}
\newcommand{\pbkPCP}[1]{\pbPCP$(#1)$}
\newcommand{\pbGPCP}{\textup{GPCP}}
\newcommand{\pbkGPCP}[1]{\pbGPCP$(#1)$}

\newcommand{\pbaccess}{\textsc{Accessibility}}
\newcommand{\pbkaccess}[1]{\textsc{Accessibility}$(#1)$}
\newcommand{\delimf}{f} 

\newcommand{\seg}[2]{[ #1, #2 ]}

\newcommand{\calak}{\mathcal{C}_k}
\newcommand{\ze}{\mathtt{0}}

\newcommand{\on}{\mathtt{1}}
\newcommand{\tw}{\mathtt{d}}

\newcommand{\tta}{\mathtt{a}}
\newcommand{\ttb}{\mathtt{b}}
\newcommand{\ab}{\left\{ \tta, \ttb \right\}}
\newcommand{\ttc}{\mathtt{c}}
\newcommand{\tte}{\mathtt{e}}
\newcommand{\zeon}{\left\{ \ze, \on \right\}}

\newcommand{\hatsigma}{\widehat{\sigma}}
\newcommand{\hatSigma}{\widehat{\Sigma}}
\newcommand{\hattau}{\widehat{\tau}}
\newcommand{\hatalpha}{\widehat{\mu}}
\newcommand{\hatT}{\widehat{T}}

\newcommand{\Thuearrow}{\pmb{\longmapsto}}
\newcommand{\deriv}[1]
{\mathrel{{\Thuearrow}_{#1}}}
\newcommand{\Deriv}[1]
{\mathrel{\underset{#1}{\Thuearrow}}}
\newcommand{\Derivstar}[1]
{\mathrel{\underset{#1}{\overset{\star}{\Thuearrow}}}}
\newcommand{\derivstar}[1]{\mathrel{{\overset{\star}{\Thuearrow}}_{#1}}}
\newcommand{\nderivstar}[1]{ \mathrel{\overset{\star}{\slashed{\Thuearrow}}_{#1}}}
\newcommand{\mv}{\varepsilon}

\newcommand{\defeq}{\mathrel{\mathop:}=}

\newcommand{\lgr}[1]{\left| #1 \right|}
\newcommand{\nocc}[2]{\left| #2 \right|_{#1}}

\newcommand{\itemonlyif}{\item(\emph{only if}).}
\newcommand{\itemif}{\item(\emph{if}).}
\newcommand{\sone}{s}
\newcommand{\tone}{s'}
\newcommand{\stwo}{t}
\newcommand{\ttwo}{t'}

\begin{document}

\maketitle 
\begin{abstract}
Let \pbkPCP{k} denote the following restriction of the well-known Post Correspondence Problem \cite{Post46PCP}: 
given alphabet $\Sigma$ of cardinality $k$ and two morphisms $\sigma$, $\tau : \Sigma^\star \to \zeon^\star$,
decide whether there exists $w \in \Sigma^+$ such that $ \sigma(w) = \tau(w)$.
Let \pbkaccess{k} denote the following restriction of the accessibility problem for semi-Thue systems:
given a $k$-rule semi-Thue system $T$ and two words $u$ and $v$, decide whether $v$ is derivable from $u$ modulo $T$.
In 1980, 
Claus showed that  if \pbkaccess{k} is undecidable  then \pbkPCP{k + 4} is also undecidable \cite{Claus80}.
The aim of the paper is to present a clean, detailed proof of the statement.

We proceed in two steps, using the Generalized Post Correspondence Problem  \cite{EhrenfeuchtR81} as an auxiliary.
Let \pbkGPCP{k} denote the following restriction of \pbGPCP:
given an alphabet $\Sigma$ of cardinality $k$, 
two morphisms $\sigma$, $\tau : \Sigma^\star \to \zeon^\star$ 
and four words $\sone$, $\stwo$, $\tone$, $\ttwo \in \zeon^\star$, 
decide whether there exists $w \in \Sigma^\star$ such that $\sone \sigma(w) \stwo = \tone \tau(w) \ttwo$.
First, we prove that if \pbkaccess{k} is undecidable then \pbkGPCP{k + 2}  is  also undecidable.
Then, we prove that if \pbkGPCP{k} is undecidable  then  \pbkPCP{k + 2} is also undecidable.
(The latter result can also be found in \cite{HarjuKHandbook}.)

To date, the sharpest undecidability bounds for both \pbPCP{} and \pbGPCP{} have been deduced from Claus's result: 
since Matiyasevich and S{\'e}nizergues showed that \pbkaccess{3} is undecidable \cite{MatiyasevichS05},
\pbkGPCP{5}  and  \pbkPCP{7} are  undecidable.
\end{abstract}

\section{Introduction}

A \emph{word} is a finite sequence of \emph{letters}.
The \emph{empty word} is denoted by $\mv$.
For every word $w$, the \emph{length} of $w$ is denoted by $\lgr{w}$.
A set of words is called a \emph{language}.
Word concatenation is denoted multiplicatively.
For every language $L$, 
$L^+$ denotes the closure of $L$ under concatenation,
and $L^\star$ denotes the language $L^+ \cup \{ \mv \}$.
An \emph{alphabet} is a finite set of {letters}.
For every alphabet $\Sigma$, $\Sigma^+$ equals the set of all non-empty words over $\Sigma$,
and  $\Sigma^\star$ equals the set of all words over $\Sigma$ including the empty word.

Let $x$ and $y$ be two words.
We say that $x$ is a \emph{prefix} (resp.~\emph{suffix}) of $y$ if there exists a word $z$ such that $x z = y$ (resp.~$zx = y$).
A  prefix (resp.~{suffix}) of $y$ is called  \emph{proper} if it is distinct from $y$.
We say that $x$ \emph{occurs} in $y$ if there exists a word $z$ such that $z x$ is a prefix of $y$.
The number of occurrences of $x$ in $y$ is denoted by $\nocc{x}{y}$: $\nocc{x}{y}$ equals the number of words $z$ such that $z x$ is a prefix of $y$.

\subsection{The (Generalized) Post Correspondence Problem}

Let $\Sigma$ and $\Delta$ be alphabets.
A function $\sigma: \Sigma^\star \to \Delta^\star$ is called a \emph{morphism} if $\sigma(x y) = \sigma(x) \sigma(y)$ for every $x$, $y \in \Sigma^\star$.
Note that any morphism maps the empty word to itself.
Moreover, for every function $\sigma_1 : \Sigma \to \Delta^\star$, there exists exactly one 
morphism $\sigma: \Sigma^\star \to \Delta^\star$ such that $\sigma(a) = \sigma_1(a)$ for every $a \in \Sigma$.
Hence, although the set of all functions from $\Sigma^\star$ to $\Delta^\star$ has the power of the continuum,  
the restriction of $\sigma$ to $\Sigma$ provides a finite encoding of $\sigma$ for every morphism $\sigma : \Sigma^\star \to \Delta^\star$. 
From now on such encodings are considered as canonical.

The well-known  Post Correspondence Problem  (\pbPCP) \cite{Post46PCP} can be stated as follows: 
given an alphabet $\Sigma$ and two morphisms $\sigma$, $\tau : \Sigma^\star \to \zeon^\star$,
decide whether there exists $w \in \Sigma^+$ such that $ \sigma(w) = \tau(w)$.
For each integer $k \ge 1$, \pbkPCP{k} denotes the restriction of  \pbPCP{} to instances  $(\Sigma, \sigma, \tau)$ such that $\Sigma$ has cardinality $k$.

The Generalized Post Correspondence Problem (\pbGPCP)  \cite{EhrenfeuchtR81} is:
given an alphabet $\Sigma$, 
two morphisms $\sigma$, $\tau : \Sigma^\star \to \zeon^\star$ 
and four words $\sone$, $\stwo$, $\tone$, $\ttwo \in \zeon^\star$, 
decide whether there exists $w \in \Sigma^\star$ such that $\sone \sigma(w) \stwo = \tone \tau(w) \ttwo$.
Note that if  $\sone \stwo = \tone  \ttwo$ then $\mv$ is a feasible solution of  \pbGPCP{} on 
$(\Sigma, \sigma, \tau, \sone, \stwo, \tone, \ttwo)$, while all feasible solutions of \pbPCP{} are non-empty words.

\begin{remark} \label{rem:PCPk-GPCPk}
For every instance  $(\Sigma, \sigma, \tau)$ of  \pbPCP, 
$(\Sigma, \sigma, \tau)$ is a yes-instance of \pbPCP{} if and only if there exists $a \in \Sigma$ such that 
$(\Sigma, \sigma, \tau, \sigma(a), \mv, \tau(a), \mv)$ is a yes-instances \pbGPCP.
\end{remark}
For each integer $k \ge 1$, \pbkGPCP{k} denotes the restriction of  \pbGPCP{} to instances  $(\Sigma, \sigma, \tau, \sone, \stwo, \tone, \ttwo)$ such that $\Sigma$ has cardinality $k$.

\subsection{Semi-Thue systems}

Formally, a  \emph{semi-Thue system} is a pair $T = (\Sigma, R)$,
where $\Sigma$ is an alphabet 
and 
where $R$ is a  subset of $\Sigma^\star \times \Sigma^\star$. 
The elements of $R$ are called the \emph{rules} of $T$.
For every $x$, $y \in \Sigma^\star$, 
we say that $y$ is \emph{immediately derivable} from $x$ modulo $T$, 
and we write $x \deriv{T} y$, 
 if there exist $s$, $t$, $z$, $z' \in  \Sigma^\star$ such that 
$x = z s z'$, $y = z t z'$ and $(s, t) \in R$.
For every $u$, $v \in \Sigma^\star$, 
we say that $u$ is \emph{derivable} from $v$ modulo $T$, and we write $u \derivstar{T} v$, 
if there exist an integer $n \ge 0$ and $x_0$, $x_1$, \ldots, $x_n \in \Sigma^\star$  
such that $x_0 = u$, $x_n = v$, and $x_{i - 1} \deriv{T} x_{i}$ for every $i \in \seg{1}{n}$:
\begin{equation} \label{eq:u-x-v}
u =
 x_0 \Deriv{T} x_1 \Deriv{T} x_2 \Deriv{T} \dotsb \Deriv{T} x_n 
= v \, .
\end{equation} 
In other words, ${\derivstar{T}}$ is the reflexive-transitive closure of the binary relation ${\deriv{T}}$.
Define the \pbaccess{} problem as: 
given a semi-Thue system $T$
and two words $u$ and $v$ over the alphabet of $T$,
decide whether $u \derivstar{T} v$.
For every integer $k \ge 1$, 
define \pbkaccess{k} as the restriction of \pbaccess{} to instances $(T, u, v)$ such that $T$ has  $k$ rules.

\subsection{Decidability}

Let $k$ be a positive integer.
The decidabilities of \pbaccess{}, \pbPCP{} and \pbGPCP{} are linked by the following four facts.

\begin{fact} \label{fact:GPCP-PCP}
If \pbkGPCP{k} is decidable then  \pbkPCP{k} is decidable.
\end{fact}

\begin{fact} \label{fact:GPCP-acc}
 If \pbkGPCP{k + 2} is decidable then \pbkaccess{k} is decidable.
\end{fact}

\begin{fact} \label{fact:PCP-GPCP}
 If \pbkPCP{k + 2} is decidable then \pbkGPCP{k} is decidable.
 \end{fact}

\begin{fact}[Claus's theorem] \label{fact:PCP-acc}
If \pbkPCP{k + 4} is decidable then  \pbkaccess{k} is  decidable.
\end{fact}

Fact~\ref{fact:GPCP-PCP} follows from Remark~\ref{rem:PCPk-GPCPk},
Fact~\ref{fact:PCP-GPCP} is \cite[Theorem~3.2]{HarjuKHandbook},
and 
Fact~\ref{fact:PCP-acc} was originally stated by Claus \cite[Theorem~2]{Claus80} (see also \cite{Claus07,HarjuKK96,HarjuKHandbook}).
To our knowledge,
Fact~\ref{fact:GPCP-acc} is explicitly stated for the first time in the present paper.

\begin{remark} \label{rem:fact-deux-trois}
The conjunction of Facts~\ref{fact:GPCP-acc} and~\ref{fact:PCP-GPCP} yields Fact~\ref{fact:PCP-acc}.
\end{remark}

Since  Matiyasevich and S{\'e}nizergues have shown that \pbkaccess{3} is undecidable  \cite[Theorem~4.1]{MatiyasevichS05}, 
it follows from Fact~\ref{fact:PCP-acc}  that \pbkPCP{7} is undecidable \cite[Corollary~1]{MatiyasevichS05}.
In the same way  Fact~\ref{fact:GPCP-acc} yields that  \pbkGPCP{5} is undecidable (see also \cite[Theorem~7]{Claus07}).
Those results are the sharpest to date.
Indeed, the decidability of each of the following eight problems is  still open:
\begin{itemize}
\item
{\pbkaccess{1}},  {\pbkaccess{2}},
\item {\pbkGPCP{3}}, {\pbkGPCP{4}},  
\item {\pbkPCP{3}}, 
{\pbkPCP{4}}, {\pbkPCP{5}} and {\pbkPCP{6}}.
\end{itemize}
However,
Ehrenfeucht and Rozenberg  showed that \pbkPCP{2} and \pbkGPCP{2} are  decidable \cite{EhrenfeuchtR81} (see also \cite{EhrenfeuchtKR82,HalavaHH02}).

\subsection{Organization of the paper}

The aim of the paper is to present a clean, detailed proof of  Fact~\ref{fact:PCP-acc}.
%
We start in Section~\ref{sec:access} with some technicalities concerning  \pbaccess.
Then, Fact~\ref{fact:GPCP-acc} is proved in Section~\ref{sec:Acc-GPCP}, and 
 Fact~\ref{fact:PCP-GPCP} is proved in Section~\ref{sec:PCP-GPCP}. 

\section{More on the decidability of \pbaccess} \label{sec:access}

\begin{definition} \label{def:Ck}
The language $\left\{ \ze \on \ze^{n}  \on \ze \on : n \ge 2 \right\}$ is denoted by $C$.
For each integer $k \ge 1$, 
define $\calak$ as the set of all instances $(T, u, v)$ of \pbaccess{} such that $u$, $v \in C^\star$ and $T = (\zeon, R)$ for some $k$-element subset $R \subseteq C^+ \times C^+$.
\end{definition}

The essential properties of the gadget language $C$ are: 
$C$ is an infinite, 
binary, 
comma-free code (see Definitions~\ref{def:code} and~\ref{def:comma-free} below),
and no word in $C$ overlaps the delimiter word $\ze \ze \on \on$.

The aim of this section is to show:

\begin{proposition} \label{prop:Ck}
 For every integer $k \ge 1$, 
the general \pbkaccess{k} problem is decidable
if and only if its restriction to $\calak$ is decidable.
\end{proposition}

The idea to prove Proposition~\ref{prop:Ck} is to construct a many-one reduction based  on the following gadget transformation:

\begin{definition} \label{def:image-Thue}
Let $T = (\Sigma, R)$ be a semi-Thue system, 
let $\Delta$ be an alphabet, 
and let $\alpha: \Sigma^\star \to \Delta^\star$.
Define the image of $T$ under $\alpha$, denoted $\alpha(T)$, as the semi-Thue system over $\Delta$ with rule set $\{ (\alpha(s) , \alpha(t)) : (s, t) \in R \}$.
\end{definition}

The next two lemmas are straightforward.

\begin{lemma} \label{lem:imply-alpha}
Let $\Sigma$ and $\hatSigma$ be alphabets, 
let $T$ be a semi-Thue system over $\Sigma$,
let $\hatT$ be a semi-Thue system over $\hatSigma$, 
and let $\alpha : \Sigma^\star \to \hatSigma^\star$ be such that 
for every $x$, $y \in \Sigma^\star$,
$x \deriv{T} y$ implies $\alpha(x) \deriv{\hatT} \alpha(y)$.
For every $u$, $v \in \Sigma^\star$,
$u \derivstar{T} v$ implies $\alpha(u) \derivstar{\hatT} \alpha(v)$.
\end{lemma}

\begin{proof}
Assume that $u \derivstar{T} v$: 
there exist an integer $n \ge 0$ and $n + 1$ words $x_0$, $x_1$, \ldots, $x_n$ over $\Sigma$ such that Equation~\eqref{eq:u-x-v} holds. 
It follows
\begin{equation}  \label{eq:alpha-u-x-v}
\alpha(u) = 
\alpha(x_0) 
\Deriv{\hatT} \alpha(x_1) 
\Deriv{\hatT} \alpha(x_2) 
\Deriv{\hatT} \dotsb 
\Deriv{\hatT} \alpha(x_n)
 = \alpha(v) \, ,
\end{equation}
and thus $\alpha(u) \derivstar{\hatT} \alpha(v)$.
\end{proof}

\begin{lemma} \label{lem:deriv-morph}
In the notation of Definition~\ref{def:image-Thue}, if $\alpha$ is a morphism then for every $u$, $v \in \Sigma^\star$, 
$u \derivstar{T} v$ implies ${\alpha(u)} \derivstar{\alpha(T)} {\alpha(v)}$.
\end{lemma} 

\begin{proof}
We apply Lemma~\ref{lem:imply-alpha} with  $\hatT \defeq \alpha(T)$. 
Let $x$, $y \in \Sigma^\star$ be such that ${x} \deriv{T} y$: 
there exist $s$, $t$, $z$, $z' \in \Sigma^\star$ such that 
$x = z  s z'$, $y = z t z'$ and $(s, t) \in R$.
Since $\alpha$ is a morphism, 
$\alpha(x)$ and $\alpha(y)$ can be parsed as follows:
$\alpha(x) = \alpha(z) \alpha(s) \alpha(z')$,
$\alpha(y) = \alpha(z) \alpha(t) \alpha(z')$
and 
$(\alpha(s),  \alpha(t))$ is a rule of $\alpha(T)$.
Hence, we get that 
${\alpha(x)} \deriv{\alpha(T)} {\alpha(y)}$.
\end{proof}

\begin{definition}
 Let $(s, t)$ be a rule of some semi-Thue system: $(s, t)$ is a pair of words.
We say that $(s, t)$ is an \emph{insertion rule} if $s = \mv$.
We say that $(s, t)$ is a \emph{deletion rule} if  $t = \mv$.
A semi-Thue system is called \emph{$\mv$-free} if it has neither insertion nor deletion rule.
By extension, an instance $(T, u, v)$ of \pbaccess{} is called  \emph{$\mv$-free} if the semi-Thue system $T$ is $\mv$-free.
\end{definition} 

Note that every instance of \pbkaccess{k} that belongs to $\calak$ is $\mv$-free.
The next two gadget morphisms play crucial roles in the proofs of both Lemma~\ref{lem:semi-Thue-free} and Theorem~\ref{th:PCP-GPCP} below.
 
\begin{definition} \label{def:lambda-rho}
 Let $\Sigma$ be an alphabet and let $d$ be a letter.
Define $\lambda_d$ and $\rho_d$ as the morphisms from 
$\Sigma^\star$ 
to 
${(\Sigma \cup \{ d \})}^\star$ 
given  by: 
$\lambda_d (a) \defeq d a $ and  $\rho_d (a) \defeq a d$  for every $a \in \Sigma$.
\end{definition}

For instance, 
$\lambda_\tw(\ze \on \on \ze\on )$ 
and 
$\rho_\tw (\ze \on \on \ze\on )$ 
equal 
$\tw \ze \tw \on \tw \on \tw \ze \tw \on $ 
and
$\ze \tw \on \tw \on \tw \ze \tw \on \tw$, 
respectively. 

\begin{lemma} \label{lem:semi-Thue-free}
For every integer $k \ge 1$, 
\pbkaccess{k} is decidable if and only if the problem is decidable on $\mv$-free instances.
\end{lemma}

\begin{proof}
We present a many-one reduction from \pbkaccess{k} in its general form to \pbkaccess{k} on $\mv$-free instances.

Let $(T, u, v)$ be an instance of \pbkaccess{k}.
Let $\Sigma$ denote the alphabet $T$,
let $d$ be a symbol such that $d \notin \Sigma$, 
and 
let $\mu : \Sigma^\star \to {(\Sigma \cup \{ d \}) }^\star$ be defined by:
$\mu(w) \defeq \lambda_d(w) d = d \rho_d(w)$ for every $w \in \Sigma^\star$. 
Clearly, $(\mu(T), \mu(u), \mu(v))$ is an $\mv$-free instance  of \pbkaccess{k},
 and $(\mu(T), \mu(u), \mu(v))$  is computable from $(T, u, v)$.

It remains to check  the correctness statement:
$u \derivstar{T} v$ if and only if  ${\mu(u)} \derivstar{\mu(T)} {\mu(v)}$.


\begin{trivlist}
 \itemonlyif{}
Let $x$, $y \in \Sigma^\star$ be such that ${x} \deriv{T} {y}$: 
there exist $s$, $t$, $z$, $z' \in \Sigma^\star$ such that 
$x = z s z'$,
$y = z t z'$ and $(s, t)$ is a rule of $T$.
Clearly, 
$\mu(x)$ and $\mu(y)$ can be parsed as follows:
$\mu(x) =  \lambda_d(z) \mu(s) \rho_d(z')$,
$\mu(y) =  \lambda_d(z) \mu(t) \rho_d(z')$
and 
 $(\mu(s), \mu(t))$ is a rule of $\mu(T)$.
Hence,  we get that ${\mu(x)} \deriv{\mu(T)} {\mu(y)}$.
It now follows from Lemma~\ref{lem:imply-alpha} 
(applied with $\alpha \defeq \mu$ and $\hatT \defeq \mu(T)$)
that $u \derivstar{T} v$ implies ${\mu(u)} \derivstar{\mu(T)} {\mu(v)}$.

\itemif{}  let $\hatalpha : {(\Sigma \cup \{d \})}^\star \to \Sigma^\star$ denote the morphism defined by:
$\hatalpha(a) \defeq a$ for every $a \in \Sigma$ and $\hatalpha(d) \defeq \mv$.
It is clear that  $\hatalpha(\mu(w)) = w$ for every $w \in \Sigma^\star$, 
and thus $T = \hatalpha(\mu(T))$.
Hence, 
for every $\hat u$, $\hat v \in {(\Sigma \cup \{d \})}^\star$,
${\hat u} \derivstar{\mu(T)} { \hat v}$ 
implies 
$\hatalpha(\hat u) \derivstar{T} \hatalpha(\hat  v)$
 by Lemma~\ref{lem:deriv-morph} 
(applied with $\alpha \defeq \hatalpha$ and $T \defeq \mu(T)$). 
In particular, 
${\mu(u)} \derivstar{\mu(T)} {\mu(v)}$ 
implies 
$u = \hatalpha(\mu(u)) \derivstar{T} \hatalpha(\mu(v)) = v$.
\end{trivlist}
\end{proof} 


Given an alphabet $\Sigma$, 
a semi-Thue system $T$ over $\Sigma$,
and a subset $L \subseteq \Sigma^\star$, 
we say that $L$ is \emph{closed under derivation modulo $T$} if for every $x \in L$ and every $y \in \Sigma^\star$, 
$x \deriv{T} y$ implies $y \in L$. 
The next lemma is an \emph{ad hoc} counterpart of Lemma~\ref{lem:imply-alpha}.

\begin{lemma} \label{lem:imply-rev-alpha}
Let $\Sigma$ and $\hatSigma$ be alphabets, 
let $T$ be a semi-Thue system over $\Sigma$,
let ${\hatT}$ be a semi-Thue system over $\hatSigma$, 
and let $\alpha : \Sigma^\star \to \hatSigma^\star$ be such that:
\begin{enumerate}[$(i)$]
\item 
the range of $\alpha$ is closed under derivation modulo $\hatT$, and
\item  for every $x$, $y \in \Sigma^\star$, 
 $\alpha(x) \deriv{\hatT} \alpha(y)$ implies $x \deriv{T} y$.
\end{enumerate}
For every $u$, $v \in \Sigma^\star$,
$\alpha(u) \derivstar{\hatT} \alpha(v)$ implies $u \derivstar{T} v$.
\end{lemma}

\begin{proof}
Assume that $\alpha(u) \derivstar{\hatT} \alpha(v)$:
there exist an integer $n \ge 0$ and $n + 1$ words $\hat x_0$, $\hat x_1$, \ldots, $\hat x_n$ over $\hatSigma$ such that 
$$
\alpha(u) =
\hat x_0 \Deriv{\hatT} \hat x_1 \Deriv{\hatT} \hat x_2 \Deriv{\hatT} \dotsb \Deriv{\hatT} \hat x_n 
= \alpha(v) \, .
$$
It follows from point $(i)$ that $\hat x_i$ belongs to the range of $\alpha$ for every $i \in \seg{0}{n}$:
let $x_0 \defeq u$, 
let $x_n \defeq v$, and for each $i \in \seg{1}{n - 1}$,  
let $x_i \in \Sigma^\star$ be such that $\hat x_i = \alpha(x_i)$.
Now,
Equation~\eqref{eq:alpha-u-x-v} holds, 
and thus Equation~\eqref{eq:u-x-v} follows by point~$(ii)$.
We have thus shown that $u \derivstar{T} v$.
\end{proof}

Surprisingly,
 hypothesis~$(i)$ of Lemma~\ref{lem:imply-rev-alpha} is not disposable.
Indeed, 
 let $T = (\Sigma, R)$ be a semi-Thue system,
and let $u_0$, $v_0 \in \Sigma^\star$ be such that $u_0 \nderivstar{T} v_0$: 
a trivial choice for $R$ is the empty set.
Let $a$ be a symbol such that $a \notin \Sigma$, 
and let ${\hatT} \defeq \left( \Sigma \cup \{ a \}, R \cup \{ (u_0, a), (a, v_0) \} \right)$.
For every $x$, $y \in \Sigma^\star$,  $x \deriv{T} y$ is equivalent to $x \deriv{\hatT} y$.
However,  
$\Sigma^\star$ is not closed under derivation modulo $\hatT$, 
and $u \derivstar{\hatT} v$ does not imply  $u \derivstar{T} v$ for every $u$, $v \in \Sigma^\star$,
since  $u_0 \derivstar{\hatT}  v_0$.

\begin{definition} \label{def:code}
Let  $X$ be a language.
We say that $X$ is a \emph{code} if the property
$$
x_1 x_2 \dotsm x_m  = y_1 y_2  \dotsm y_n
\iff 
(x_1, x_2, \dotsc, x_m) = (y_1, y_2, \dotsc, y_n)
$$ 
holds for any integers $m$, $n \ge 1$ and any elements $x_1$, $x_2$, \ldots, $x_m$, $y_1$, $y_2$, \ldots, $y_n \in X$.
\end{definition} 

Note that $(x_1, x_2, \dotsc, x_m) = (y_1, y_2, \dotsc, y_n)$ means that both $m = n$ and $x_i = y_i$ for every $i \in \seg{1}{m}$.
In other words, a language $X$ is a code if each word in $X^\star$ has a unique factorization over $X$.
A morphism  $\alpha: \Sigma^\star \to \Delta^\star$ is injective if and only if 
$\alpha$ is injective on $\Sigma$ and $\alpha(\Sigma)$ is a code.

\begin{definition}[\cite{BerstelP85}] \label{def:comma-free}
A code $X$ is called \emph{comma-free} if for every words $x$, $z$ and $z'$,
$$(x \in X \mathrel{\textup{and}}  z  x z' \in X^\star ) \implies (z \in  X^\star \mathrel{\textup{and}} z' \in X^\star) \, .
$$
\end{definition}

Every comma-free code is a \emph{bifix} code: 
no  word in the language is a  prefix or a suffix of another word in the language. 
For instance, 
$K \defeq  \left\{ \on \ze^n \on : n \ge 1 \right\}$ and $C$ are  comma-free codes,
but $K \cup \{ \on \on \}$ is a bifix code which is not comma-free.

\begin{lemma} \label{lemm:image-comma-free}
In the notation of Definition~\ref{def:image-Thue}, assume that
\begin{enumerate}[$(i)$]
\item $\alpha$ is an injective morphism,
\item $\alpha(\Sigma)$ is a comma-free code, and 
\item $T$ has no insertion rule.
\end{enumerate}
For every $u$, $v \in \Sigma^\star$, 
$u \derivstar{T} v$ is equivalent to 
$\alpha(u) \derivstar{\alpha(T)} \alpha(v)$.
\end{lemma}

\begin{proof}
According to Lemma~\ref{lem:deriv-morph},
 $u \derivstar{T} v$ implies $\alpha(u) \derivstar{\alpha(T)} \alpha(v)$.
Conversely, let us prove that $\alpha(u) \derivstar{\alpha(T)} \alpha(v)$ implies $u \derivstar{T} v$. We rely on Lemma~\ref{lem:imply-rev-alpha}.

Let $\hat x$, $\hat y \in \Delta^\star$
 be such that $\hat x$ belongs to the range of $\alpha$ and $\hat x \deriv{\alpha(T)} \hat y$:
there exist $\hat z$, $\hat z' \in \Delta^\star$ and $(s, t) \in R$ such that 
$\alpha(x) = \hat z \alpha(s) \hat z'$ and 
$\hat y = \hat z \alpha(t) \hat z'$. 
Since $\alpha$ is a morphism, the range of $\alpha$ equals ${\alpha(\Sigma)}^\star$.
In particular, 
both $\alpha(s)$ and $\hat z \alpha(s) \hat z'$ belong to ${\alpha(\Sigma)}^\star$.
Furthermore, 
$\alpha(s)$ belongs to ${\alpha(\Sigma)}^+$: 
indeed,
$s$ is a non-empty word because $T$ has no insertion rule,
and thus its image under the injective morphism $\alpha$ is also  a non-empty word.
It follows that  both $\hat z$ and  $\hat z'$ belong to ${\alpha(\Sigma)}^\star$ because $\alpha(\Sigma)$ is a comma-free code: 
there exist $z$, $z' \in \Sigma^\star$ such that $\alpha(z) = \hat z$ and $\alpha(z') = \hat z'$.
We can now write $\hat x$ and $\hat y$ in the forms $\hat x  = \alpha(z s z')$ and  $\hat y = \alpha( z t z')$.
Hence, $\hat y$ belongs to the range of $\alpha$,
which proves that the range of $\alpha$ is closed under derivation modulo $\alpha(T)$.
Moreover, we also get $\alpha^{-1}(\hat x) = z s z' \deriv{T} z t z' = \alpha^{-1}(\hat y)$.
Therefore, Lemma~\ref{lem:imply-rev-alpha} applies with $\hatT \defeq \alpha(T)$.
\end{proof}

Let us thoroughly examine the hypotheses of Lemma~\ref{lemm:image-comma-free}.
Hypothesis~$(iii)$ could be replaced with ``\emph{$T$ has no deletion rule}'': apply Lemma~\ref{lemm:image-comma-free} in its original form to the reversal of $T$, 
defined as $\widetilde{T} \defeq \left(\Sigma, \{ (t, s) : (s, t) \in R \} \right)$.
However, the following two counterexamples show that neither hypothese~$(ii)$ nor hypothese~$(iii)$  is disposable.

\begin{cex} \label{cex:disprove-inj}
Let $T \defeq (  \ab, \{ (\tta, \tta \tta) \})$ and let $\alpha : \ab^\star \to \zeon^\star$ be the morphism defined by 
$\alpha(\tta) \defeq \ze \on $ 
and 
$\alpha(\ttb) \defeq \ze \on \on $: 
$\alpha(T) = (\zeon, \{ (\ze \on, \ze \on \ze \on) \})$.
Clearly, $\alpha$ is injective and $T$ is $\mv$-free.
However,  $\alpha(\ab)$ is not a comma-free code, and  
$\alpha(u)  \derivstar{\alpha(T)} \alpha(v)$
does not imply 
 $u \derivstar{T}  v$  for every $u$, $v \in \ab^\star$:
$\alpha(\ttb)  \deriv{\alpha(T)} \alpha(\tta \ttb)$ but ${\ttb} \nderivstar{T}  {\tta \ttb}$.
\end{cex}

Counterexample~\ref{cex:disprove-inj}
 disproves a claim from Claus's original paper \cite[page~$57$, line~$-4$]{Claus80}.
A statement from Harju, Karhum{\"a}ki and Krob \cite[page~$43$, line~$1$]{HarjuKK96} is disproved in the same way.

\begin{cex} \label{ex:disprove-comma}
Let 
$T \defeq \left( \{ \tta, \ttb, \ttc \}, \{ (\mv, \tta), (\ttb, \mv) \} \right)$
and 
let 
$\alpha :  {\{ \tta, \ttb, \ttc \}}^\star \to \zeon^\star$ 
be the morphism defined by 
$\alpha(\tta) \defeq \on \ze \on$, 
$\alpha(\ttb) \defeq \on \ze \ze \on$
and 
$\alpha(\ttc) \defeq \on \ze \ze \ze \on$: 
$\alpha(T) = \left(\zeon, \{ (\mv, \on \ze \on), (\on \ze \ze \on, \mv) \} \right)$.
Clearly, $\alpha$ is injective and $\alpha(\{ \tta, \ttb, \ttc \})$ is a comma-free code.
However, $T$ admits both insertion and deletion rules, 
and $\alpha(u)  \derivstar{\alpha(T)} \alpha(v)$
does not imply 
 $u \derivstar{T}  v$  for every $u$, $v \in \ab^\star$:
$\ttc \nderivstar{T} \tta$ 
but
$$
\alpha(\ttc)
\Deriv{\alpha(T)}
\on \ze \on \ze \on \ze \ze \on
\Deriv{\alpha(T)} 
\on \ze \on \ze \on \ze \ze \on \ze \on \on 
\Deriv{\alpha(T)}
\on \ze \on \ze  \ze \on \on
\Deriv{\alpha(T)} 
\alpha(\tta) \, .
$$
\end{cex}

\begin{proof}[Proof of Proposition~\ref{prop:Ck}]
We present a many-one reduction 
from \pbkaccess{k} on $\mv$-free instances 
to \pbkaccess{k} on $\calak$, 
so that Lemma~\ref{lem:semi-Thue-free} applies.

Let $(T, u, v)$ be an $\mv$-free instance of \pbkaccess{k}.
Let $\Sigma$ denote the alphabet of $T$.
Compute an injection $\alpha : \Sigma \to C$.
The morphism from $\Sigma^\star$ to $\zeon$  that extends $\alpha$ is also denoted $\alpha$.
Clearly, $(\alpha(T), \alpha(u), \alpha(v))$ belongs to $\calak$,
and $(\alpha(T), \alpha(u), \alpha(v))$ is computable from $(T, u, v)$.
Moreover, $u \derivstar{T} v$ is equivalent to $\alpha(u) \derivstar{\alpha(T)} \alpha(v)$ by Lemma~\ref{lemm:image-comma-free}.
\end{proof}

\section{From \pbGPCP{} to \pbaccess} \label{sec:Acc-GPCP}

The key ingredient of the proof of Fact~\ref{fact:GPCP-acc} is the \emph{accordion lemma} (Lemma~\ref{lem:accordion} below).

\begin{definition} 
A word $f$ is called \emph{bordered} if there exist three non-empty words $u$,  $v$ and $x$ such that $f = x u = vx$.
\end{definition} 

Equivalently,
$f$ is bordered if and only if there exists a word $z$  with $\lgr{z}  <  2\lgr{f}$ such that $f$ occurs twice or more in $z$.

\begin{definition} 
We say that two words $x$ and $y$ \emph{overlap} if at least one of the  following four assertions hold: 
\begin{enumerate}[$(i)$]
\item
$x$ occurs in $y$, 
\item 
$y$ occurs in $x$, 
\item 
 some non-empty prefix of $x$ is a suffix of $y$, or 
\item 
 some non-empty prefix of $y$ is a suffix of $x$.
 \end{enumerate} 
\end{definition} 

Since we make the convention that the empty word occurs in every word, 
the empty word and any other word do overlap.
Two non-empty words $x$ and $y$ overlap if and only if there exists a word $z$ with $\lgr{z} < \lgr{x} + \lgr{y}$  such that both $x$ and $y$ occur in $z$.

We can now  state a protoversion of the accordion lemma.

\begin{lemma} \label{lem:proto-accordion}
Let $T = (\Sigma, R)$ be a semi-Thue system, 
and 
let $f$, $u$, $v \in \Sigma^\star$  be  such that:
\begin{enumerate}[$(i)$]
\item $f$ is unbordered,
\item $f$ does not occur in $u$, 
\item $f$ does not occur in $v$, and
\item  for each rule $(s, t) \in R$, $s$ and $f$ do not overlap.
\end{enumerate} 
Then, $u \derivstar{T} v$
holds if and only if there exist $x$, $y \in {\Sigma}^\star$ satisfying both 
$x f  v = u f y$ 
and 
$x \derivstar{T} y$.
\end{lemma}

\begin{proof}
\begin{trivlist}
\itemonlyif{}
If $u \derivstar{T} v$ then $x \defeq  u$ and $y \defeq v $ are such that
$x f  v  =  u f y$ 
and 
$x \derivstar{T} y$.
\itemif{} 
Assume that  there exist $x$, $y \in {\Sigma}^\star$ such that $x f  v  =  u f y$ and $x \derivstar{T} y$.
Let $n$ denote the number of occurrences of $f$ in $x f v$.
Since $f$ is unbordered (hypothesis~$(i)$), 
those occurrences  are pairwise non-overlapping: 
$$
x f v = u f y = w_0 f w_1 f w_2 \dotsm f w_n
$$  
for some words $w_0$, $w_1$, \ldots, $w_n \in \Sigma^\star$.
Since $f$ does not occur in $v$ (hypothesis~$(iii)$),
we have  $v = w_n$ and 
$$
x  =  w_0 f w_1 f w_2  \dotsm f w_{n - 1} \, .
$$
In the same way, $f$ does not occur in $u$ either (hypothesis~$(ii)$), and thus
we have also $u = w_0$ and 
$$
y  = w_1 f w_2 f  w_3  \dotsm f w_n \, .
$$
Now, remark that $f$ plays the role of a delimiter with respect to the derivation modulo $T$ (hypothesis~$(iv)$):
$x \derivstar{T} y$ 
implies 
$$
{w_{i - 1}} 
\Derivstar{T}
{w_i}
$$
for every $i \in \seg{1}{n}$.
 Therefore, $u \derivstar{T} v$ holds.
 \end{trivlist}
\end{proof}

Let us comment the statement of  Lemma~\ref{lem:proto-accordion}. 
Hypothesis~$(iv)$ implies that $T$ has no insertion rule.
It could be replaced with ``\emph{for every $(s, t) \in R$, $t$ and $f$ do not overlap}''. 
Hypothesis~$(ii)$ is in fact disposable: 
the verification is left to the reader.
Let $\Sigma$ be an alphabet, 
let $f$ be a symbol such that $f \notin \Sigma$,
 and let 
 $T$ be a semi-Thue system over $\Sigma \cup \{ f \}$ with rules in $\Sigma^+ \times \Sigma^\star$.
An easy consequence of Lemma~\ref{lem:proto-accordion} is that,
for every $u$, $v \in \Sigma^\star$, 
$u \derivstar{T} v$
holds if and only if there exist $x$, $y \in {(\Sigma \cup \{ f \})}^\star$ satisfying both 
$x f v = u f y$ 
and 
$x \derivstar{T} y$.

\begin{definition} \label{def:f}
The word $\ze \ze \on \on$ is denoted by $\delimf$.
\end{definition}

\begin{lemma}[Accordion lemma] \label{lem:accordion}
Let $k$ be a positive integer.
For every $(T, u, v) \in \calak$, 
$(T, u, v)$ is a yes-instance of \pbkaccess{k} 
if and only if
there exist $x$, $y \in {\zeon}^\star$ satisfying both 
$x \delimf v = u  \delimf y$ and  $x \derivstar{T} y$.
\end{lemma} 

\begin{proof}
Clearly, $\delimf$ is unbordered,
 and for every $s \in C^+$, $s$ and $\delimf$ do not overlap.
Hence, Lemma~\ref{lem:proto-accordion} applies.
\end{proof} 

The statement of the accordion lemma can be made precise as follows (the verification is left to the reader):
for any  $(T, u, v) \in \calak$ and any  $x$, $y \in \zeon^\star$ such that $x \delimf v = u  \delimf y$ and  $x \derivstar{T} y$, both $x$ and $y$ belong to ${(C \cup \{ \delimf \})}^\star$.
Besides, if $C$ and $\delimf$ were  defined as $C \defeq \left\{ \on  \ze^{n}   \on : n \ge 1 \right\}$ and $\delimf \defeq \on \on$ in Definitions~\ref{def:Ck} and~\ref{def:f}, then Proposition~\ref{prop:Ck} and Lemma~\ref{lem:accordion} would hold (the verification is left to the reader).
In \cite[Theorem~4.1]{HarjuKHandbook}, Harju and Karhum\"aki present a proof of Claus's theorem
that implicitly relies  on those variants of  Proposition~\ref{prop:Ck} and Lemma~\ref{lem:accordion}.

\begin{definition}
An instance  $(\Sigma, \sigma, \tau, \sone, \stwo, \tone, \ttwo)$ of \pbGPCP{} is called \emph{erasement-free} if $\sigma(\Sigma) \cup \tau(\Sigma) \subseteq \zeon^+$.
\end{definition}

We can now prove a slightly strengthened version of Fact~\ref{fact:GPCP-acc}.

\begin{theorem} \label{th:GPCP-thue}
Let $k$ be a positive integer.
If \pbkGPCP{k + 2} is decidable on  erasement-free instances  then \pbkaccess{k} is decidable.
\end{theorem}

\begin{proof} 
In order to apply Proposition~\ref{prop:Ck}, 
we present a many-one reduction from \pbkaccess{k} on $\calak$ to \pbkGPCP{k + 2} on erasement-free instances. 

Let $(T, u, v)$ be an element of $\calak$:
there exist $s_1$, $s_ 2$, \ldots, $s_k$, $t_1$, $t_2$, \ldots, $t_k \in C^+$ such that
$T = \left(\zeon, \{ (s_1, t_1), (s_2, t_2), \dotsc, (s_k, t_k)  \} \right)
$. 

Let $a_1$, $a_2$, \ldots, $a_k$ be $k$ symbols such that
 $\Sigma \defeq  \{ \ze, \on, a_1, a_2, \dotsc, a_k \}$ is an alphabet of cardinality $k + 2$.
Let $\sigma$, $\tau : \Sigma^\star \to \zeon^\star$ be the morphisms defined by:
\begin{align*}
\sigma(\ze) & \defeq \ze \,, & \tau(\ze) & \defeq \ze \, , \\ 
\sigma(\on) & \defeq \on \,, & \tau(\on) & \defeq \on \, , \\ 
\sigma(a_i) & \defeq s_i \,, & \tau(a_i) & \defeq t_i  
\end{align*}
for every $i \in \seg{1}{k}$.
Let $J$ denote the instance $(\Sigma, \sigma, \tau, \mv, \delimf v, u \delimf , \mv)$ of \pbkGPCP{k + 2}.

It is clear that $J$ is erasement-free and that $J$ is  computable from $I$.
It remains to check that 
$I$ is a yes-instance of \pbkaccess{k}
if and only if 
$J$ is a yes-instance of \pbkGPCP{k + 2}.
The proof of the ``if part'' relies on the accordion lemma while the proof of the ``only if part'' relies on next lemma.

\begin{lemma} \label{lem:sigma-tau-xyz}
For any $x$, $y \in \zeon^\star$ such that $x \deriv{T} y$, 
there exists $z \in \Sigma^\star$ such that $x = \sigma(z)$ and $y = \tau(z)$. 
\end{lemma}

\begin{proof}
Let $z'$, $z'' \in \zeon^\star$ and let $i \in \seg{1}{k}$ be such that $x = z' s_i z''$ and $y = z' t_i z''$.
A suitable choice for $z$ is $z' a_i z''$.
\end{proof}

\begin{trivlist}
\itemonlyif{} 
Assume that $u \derivstar{T} v$.
There exist an integer $n \ge 0$ and $n + 1$ words $x_0$, $x_1$, \ldots, $x_n$ over $\zeon$ such that Equation~\eqref{eq:u-x-v} holds.
Lemma~\ref{lem:sigma-tau-xyz} ensures that  there exists $z_i \in \Sigma^\star$ satisfying  $x_{i - 1} = \sigma(z_i)$ and $x_i = \tau(z_i)$ for each $i \in \seg{1}{n}$.
Now, $w \defeq z_1 \delimf z_2 \delimf z_3 \dotsm \delimf z_{n}$ is such that 
$\sigma(w) \delimf v = u \delimf  \tau(w)$.
\itemif{} Assume that there exists $w \in \Sigma^\star$ such that $\sigma(w) \delimf v = u \delimf  \tau(w)$.
The morphisms  $\sigma$ and $\tau$ are defined in such a way that  ${\sigma(z)} \derivstar{T} {\tau(z)}$ for every $z \in \Sigma^\star$.
In particular, $x := \sigma(w)$ and $y := \tau(w)$ are such that  $x \delimf v = u\delimf y$ and $x \derivstar{T} y$.
Hence, Lemma~\ref{lem:accordion} yields  $u \derivstar{T} v$.
\end{trivlist}
\end{proof}

Combining Theorem~\ref{th:GPCP-thue} and \cite[Theorem~4.1]{MatiyasevichS05} we obtain:

\begin{corollary}
\label{coro:GPCPcinq}
\pbkGPCP{5} is undecidable on erasement-free instances.
\end{corollary}

\section{From \pbPCP{} to \pbGPCP} \label{sec:PCP-GPCP}

\begin{definition}
An instance $(\Sigma, \sigma, \tau, \sone, \stwo, \tone, \ttwo)$  of \pbGPCP{} is called \emph{$(\mv, \mv)$-free} if for every $a \in \Sigma$, $(\sigma(a), \tau(a))  \ne (\mv, \mv)$.
\end{definition}

\begin{lemma} \label{lem:mv-mv-free}
For every integer $k \ge 1$, 
\pbkGPCP{k} is decidable if and only if the problem is decidable on $(\mv, \mv)$-free instances.
\end{lemma} 

\begin{proof}
We present a many-one reduction from \pbkGPCP{k} to \pbkGPCP{k} on $(\mv, \mv)$-free instances.

Let $I \defeq (\Sigma, \sigma, \tau, \sone, \stwo, \tone, \ttwo)$ be an instance of \pbkGPCP{k}.
Compute the set $\hatSigma$ of all letters $a \in \Sigma$ such that $(\sigma(a), \tau(a))  \ne (\mv, \mv)$.
If $\hatSigma$ is empty then solving  \pbkGPCP{k} on $I$ reduces to checking whether $\sone \stwo$ and $\tone \ttwo$
are equal.
Hence, we may assume   $\hatSigma \ne \emptyset$ without loss of generality, 
taking out of the way cumbersome considerations.
Let $\hatsigma$ and $\hattau$ denote the restrictions to $\hatSigma^\star$ of $\sigma$  and $\tau$, respectively.
Let $J$ denote the septuple $(\hatSigma,  \hatsigma,  \hattau, \sone, \stwo, \tone, \ttwo)$.
Clearly, $J$ is an $(\mv, \mv)$-free instance of \pbkGPCP{k} and $J$ is computable from $I$.
Moreover,  $I$ is a yes-instance of \pbkGPCP{k} if and only if $J$ is also a yes-instance of the problem.
\end{proof}

Remark that every erasement-free instance of \pbGPCP{} is $(\mv, \mv)$-free, but the converse is false in general.

\begin{definition} 
An instance $(\Sigma, \sigma, \tau)$ of \pbPCP{} is called \emph{erasement-free} if $\sigma(\Sigma) \cup \tau(\Sigma) \subseteq \zeon^+$.
\end{definition}

We can now prove Fact~\ref{fact:PCP-GPCP}.

\begin{theorem}
\label{th:PCP-GPCP}
Let $k$ be a positive integer.
\begin{enumerate}[$(i)$.]
\item If \pbkPCP{k + 2} is decidable then  \pbkGPCP{k} is decidable.
\item 
 If \pbkPCP{k + 2} is decidable on erasement-free instances then  \pbkGPCP{k} is decidable on erasement-free instances.
 \end{enumerate}
\end{theorem}

\begin{proof}
We present a many-one reduction from \pbkGPCP{k} on $(\mv, \mv)$-free instances to \pbkPCP{k + 2} in order to apply Lemma~\ref{lem:mv-mv-free}.

Let $I \defeq (\Sigma, \sigma, \tau, \sone, \stwo, \tone, \ttwo)$ be an $(\mv, \mv)$-free
instance of \pbkGPCP{k}.
Without loss of generality, we may assume $\ttb \notin \Sigma$ and $\tte \notin \Sigma$: 
$\hatSigma \defeq \Sigma \cup \{ \ttb, \tte \}$ is an alphabet of cardinality $k + 2$.
Let  $\lambda \defeq \lambda_\tw$ and $\rho \defeq \rho_\tw$ (see Definition~\ref{def:lambda-rho}).
Let $\hatsigma$, $\hattau : {\hatSigma}^\star \to {\{ \ze, \on, \tw, \ttb, \tte \}}^\star$ be the two morphisms defined by:
\begin{align*}
\hatsigma(\ttb) & \defeq \ttb \lambda(\sone) \,, &
\hattau(\ttb) & \defeq \ttb \tw \rho(\tone)\, , \\ 
\hatsigma (\tte) & \defeq \lambda(\stwo) \tw \tte \, , &
\hattau(\tte) & \defeq  \rho(\ttwo) \tte \, ,  \\
\hatsigma (a) & \defeq \lambda(\sigma(a))  \,, & 
\hattau(a) & \defeq  \rho(\tau(a)) 
\end{align*}
for every $a \in \Sigma$.
Let $\jmath : {\{ \ze, \on, \tw, \ttb, \tte \}}^\star \to \zeon^\star$ denote an  injective morphism: 
for instance  $\jmath$ can be given by 
$\jmath( \ze) \defeq \ze \ze \ze$, 
$\jmath( \on) \defeq \on \on \on$, 
$\jmath( \tw) \defeq \on \ze \on$,
$\jmath( \ttb) \defeq \on \ze \ze$ and 
$\jmath( \tte) \defeq \ze \ze \on$. 

It is clear that 
$J \defeq 
(\hatSigma, 
\jmath \circ \hatsigma, 
\jmath \circ \hattau)$ is an instance of \pbkPCP{k + 2} computable from $I$,
 and that $J$ is erasement-free whenever  $I$ is erasement-free.
Hence, to prove both points~$(i)$ and~$(ii)$ of Theorem~\ref{th:PCP-GPCP}, 
it remains to check that
$I$ is a yes-instance of   \pbkGPCP{k}
if and only if  
$J$ is a yes-instance of   \pbkPCP{k + 2}.

\begin{lemma} \label{lem:w-bwe}
For every $w \in \Sigma^\star$, 
$\sone \sigma(w) \stwo = \tone \tau(w) \ttwo$ if and only if 
$\hatsigma(\ttb w \tte) = \hattau(\ttb w \tte)$.
\end{lemma}

\begin{proof}
Straightforward computations yield
$$
\hatsigma(\ttb w \tte)
= 
\hatsigma(\ttb) \hatsigma( w ) \hatsigma(\tte)
= 
\ttb \lambda(\sone) \lambda(\sigma(w)) \lambda(\stwo) \tw \tte
= 
\ttb \lambda(\sone \sigma(w) \stwo) \tw \tte
$$
and 
$$
\hattau(\ttb w \tte)
= 
\hattau(\ttb) \hattau( w ) \hattau(\tte)
= 
\ttb \tw \rho(\tone) \rho(\tau(w)) \rho(\ttwo) \tw \tte
= \ttb \tw \rho(\tone \tau(w) \ttwo)  \tte \, .
$$
Since $\lambda(x) \tw = \tw \rho(x)$ for every $x \in \zeon^\star$, 
$\sone \sigma(w) \stwo = \tone \tau(w) \ttwo$ implies $\hatsigma(\ttb w \tte) = \hattau(\ttb w \tte)$,
and furthermore, 
$\hatsigma(\ttb w \tte) = \hattau(\ttb w \tte)$ implies $\lambda(\sone \sigma(w) \stwo) =  
\lambda(\tone \tau(w) \ttwo)$.
Since $\lambda$ is trivially injective, 
$\hatsigma(\ttb w \tte) = \hattau(\ttb w \tte)$ implies $\sone \sigma(w) \stwo = \tone \tau(w) \ttwo$.
\end{proof}

 If $I$ is a yes instance of \pbkGPCP{k} then it follows from Lemma~\ref{lem:w-bwe} that $J$ is a yes-instance of \pbkPCP{k + 2}.
The converse is slightly more complicated to prove.

\begin{lemma} \label{lem:prefix-we}
For every  $w \in {\hatSigma}^\star$, the following three assertions are equivalent:
\begin{enumerate}
\item $\hatsigma( w \tte)$ is a prefix of  $\hattau( w \tte)$,
\item $\hattau( w \tte)$ is a prefix of $\hatsigma( w \tte)$, and 
\item $\hatsigma( w \tte) = \hattau( w \tte)$.
\end{enumerate}
\end{lemma}

\begin{proof}
The letter $\tte$ occurs once in $\hatsigma( \tte)$ 
(resp.~$\hattau( \tte)$) 
whereas for every $a \in \Sigma \cup \{ \ttb \}$,
$\tte$ does not occur at all in $\hatsigma(a)$ (resp.~$\hattau( a)$).
Therefore,  $\nocc{\tte}{\hatsigma( x)} = \nocc{\tte}{x}  = \nocc{\tte}{\hattau(x)}$ holds for every $x \in {\hatSigma}^\star$. 
Since $\tte$ is  the last letter of $\hatsigma( \tte)$,
any proper prefix of $\hatsigma( w \tte)$  contains less occurrences of $\tte$ than $\hattau(w\tte)$.
 From that we deduce that $\hattau( w \tte)$ 
cannot be a proper prefix of  $\hatsigma( w \tte)$.
In the same way,  $\hatsigma( w \tte)$ cannot be a proper prefix of $\hattau( w \tte)$.
\end{proof} 

\begin{lemma} \label{lem:suffix-bw}
For every  $w \in {\hatSigma}^\star$, the following three assertions are equivalent:
\begin{enumerate}
\item $\hatsigma(\ttb w)$ is a suffix of  $\hattau(\ttb  w )$, 
\item $\hattau(\ttb w )$ is a suffix of $\hatsigma(\ttb w )$, and 
\item $\hatsigma( \ttb w) = \hattau(\ttb  w )$. 
\end{enumerate}
\end{lemma}

\begin{proof}
Lemma~\ref{lem:suffix-bw} is proved in the same way as Lemma~\ref{lem:prefix-we}.
The details are left to the reader.
\end{proof} 

\begin{claim} \label{claim:hatsigma}
Let $a \in \hatSigma$ be such that $\hatsigma(a) \ne \mv$.
\begin{enumerate}[$(i)$.]
\item The first letter of $\hatsigma(a)$ is either $\ttb$ or $\tw$.
\item The last letter of  $\hatsigma(a)$ is distinct from $\tw$. 
\end{enumerate}
\end{claim}

\begin{claim} \label{claim:hattau}
Let $a \in \hatSigma$ be such that $\hattau(a) \ne \mv$.
\begin{enumerate}[$(i)$.]
\item The first letter of $\hattau(a)$ is distinct from $\tw$.
\item The last letter of  $\hattau(a)$ is either $\tw$ or $\tte$.
\end{enumerate}
\end{claim}

Assume that $J$ is a yes-instance of \pbkPCP{k + 2}.
Let $w \in {\hatSigma}^+$ be such that $\hatsigma(w) = \hattau( w )$.
Let $x$ denote both words $\hatsigma(w)$ and $\hattau( w )$. 

Since $I$ is an $(\mv, \mv)$-free instance of \pbGPCP, $(\hatsigma(a), \hattau(a))$ is distinct from $(\mv,  \mv)$ for every $a \in \hatSigma$, 
and thus $x$ is a non-empty word.
Combining Claims~\ref{claim:hatsigma}$(i)$ and~\ref{claim:hattau}$(i)$,
we obtain that $\ttb$ is the first letter of $x$, 
and thus $\ttb$ is also the first letter of $w$.
In the same way, combining Claims~\ref{claim:hatsigma}$(ii)$ and~\ref{claim:hattau}$(ii)$, 
we obtain that $\tte$ is the last letter of $x$, 
and thus $\tte$ is also the last  letter of $w$.
Hence, $w$ is of the form $\ttb w' \tte$ with $w' \in {\hatSigma}^\star$.

Now, assume that $w$ is a \emph{shortest} non-empty word over $\hatSigma$ such that $\hatsigma(w) = \hattau( w)$. 
Let us check that $w' \in \Sigma^\star$.
By the way of contradiction suppose that $\tte$ occurs in $w'$: 
there exist $w_1$, $w_2 \in {\hatSigma}^\star$ such that $w' = w_1 \tte w_2$. 
Straightforward computations yield $
\hatsigma(\ttb w_1 \tte) \hatsigma(w_2 \tte)  = 
x = \hattau( \ttb w_1 \tte)  \hattau(w_2 \tte) 
$. 
Therefore,
 $\hatsigma( \ttb w_1 \tte)$ is a prefix of $\hattau(\ttb w_1 \tte)$
 or 
$\hattau( \ttb w_1 \tte)$ is a prefix of $\hatsigma(\ttb w_1 \tte)$.
From Lemma~\ref{lem:prefix-we}, we deduce that $\hatsigma( \ttb w_1 \tte) = \hattau(\ttb w_1 \tte)$.
Since $\ttb w_1 \tte$ is shorter than $w$, a contradiction follows.
Hence $\tte$ does not occur in $w'$. 
Similar arguments based on Lemma~\ref{lem:suffix-bw}  show that $\ttb$ does not occur in $w'$ either.

Hence,  $w'$ is a word over $\Sigma$,
 and thus Lemma~\ref{lem:w-bwe} ensures that $\sone \sigma(w') \stwo = \tone \tau(w') \ttwo$.
It follows that $I$ is a yes-instance of \pbkGPCP{k}. 
\end{proof} 

Strictly speaking, the correspondence problem that was originally introduced by Post in his 1946 paper \cite{Post46PCP} is, in our terminology, the restriction of \pbPCP{} to erasement-free instances.

Combining Theorems~\ref{th:GPCP-thue} and~\ref{th:PCP-GPCP}$(ii)$, 
we obtain a slightly strengthened version of Claus's theorem (Fact~\ref{fact:PCP-acc}).

\begin{corollary}
\label{coro:Claus}
Let $k$ be a positive integer.
If \pbkPCP{k + 4} is decidable on  erasement-free instances  then \pbkaccess{k} is decidable.
\end{corollary}

Combining  Corollary~\ref{coro:Claus} and \cite[Theorem~4.1]{MatiyasevichS05} we obtain:

\begin{corollary}
\pbkPCP{7} is undecidable on erasement-free instances.
\end{corollary}

\bibliography{gpcp}
\bibliographystyle{plain}

\end{document}